%% file: exponent.tex
\documentclass[journal]{IEEEtran}
\usepackage{amssymb,stmaryrd,amsmath,amsfonts,rotating}%amsfonts,amsthm
\usepackage[noadjust]{cite}
\usepackage{color}%temporary packages
\usepackage{epic}
\RequirePackage{bbm}
\input{./definition}

\allowdisplaybreaks
\begin{document}
\title{Polar Codes: Characterization of Exponent, Bounds, and Constructions}
\author{Satish Babu Korada, Eren \c Sa\c so\u glu and R{\"u}diger Urbanke}

\maketitle
\begin{abstract}
Polar codes were recently introduced by Ar\i kan. They
achieve the capacity of arbitrary symmetric binary-input discrete memoryless channels
under a low complexity
successive cancellation decoding strategy. 
The original polar code construction is closely related to the
recursive construction of Reed-Muller codes and is based on the $2 \times 2$  matrix 
$\bigl[ \begin{smallmatrix} 1 &0 \\ 1& 1 \end{smallmatrix} \bigr]$.
It was shown by Ar\i kan and Telatar that this construction achieves an error exponent of $\frac12$, i.e.,
that for sufficiently large blocklengths the error probability decays 
exponentially in the square root of the length.
It was already mentioned by Ar\i kan that in principle larger matrices can be used to construct polar codes. 
A fundamental question
then is to see whether there exist matrices with exponent exceeding $\frac12$.
We first show that any $\ell \times \ell$ matrix none of whose 
column permutations is upper triangular polarizes symmetric channels.
We then characterize the exponent of a given square matrix and derive upper and
lower bounds on achievable exponents. Using these bounds we show that
there are no matrices of size less than $15$ with exponents exceeding $\frac12$. Further, we give
a general construction based on BCH codes which for large $n$ achieves exponents arbitrarily close to $1$
and which exceeds $\frac12$ for size $16$.
\end{abstract}
\section{Introduction}
Polar codes, introduced by Ar\i kan in \cite{Ari08}, are the first provably capacity achieving codes for arbitrary symmetric binary-input discrete memoryless channels (B-DMC) with low encoding and decoding complexity. The polar code construction is based on the following observation: Let 
\begin{align}\label{eqn:2by2}
G_2=\left[
\begin{array}{cc}
1 & 0 \\
1 & 1
\end{array}
\right].
\end{align}
Apply the transform $G_2^{\otimes n}$ 
(where ``$\phantom{}^{\otimes n}$'' denotes the $n^{th}$ Kronecker power) to a block of 
$N = 2^n$ bits and transmit the output through independent copies of a B-DMC $W$ (see Figure \ref{fig:transform}). 
 As $n$ grows large, the channels seen by individual bits (suitably defined in \cite{Ari08}) start \emph{polarizing}: they approach either a noiseless channel or a pure-noise channel, where the fraction of channels becoming noiseless is close to the symmetric mutual information $I(W)$.

It was conjectured in \cite{Ari08} that polarization is a general phenomenon,
and is not restricted to the particular transformation $G_2^{\otimes n}$. In
this paper we first give a partial affirmation to this conjecture. In particular, we
consider transformations of the form $G^{\otimes n}$ where $G$ is an
$\ell\times\ell$ matrix for $\ell \geq 3$ and provide necessary and sufficient
conditions for such $G$s to polarize symmetric B-DMCs. 

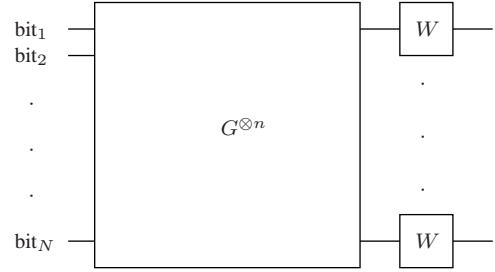
\begin{figure}[ht]
\centering
\input{ps/transform}
\caption{The transform $G^{\otimes n}$ is applied and the resulting vector
is transmitted through the channel $W$.}\label{fig:transform}
\end{figure}

For the matrix $G_2$ it was shown by Ar\i kan and Telatar \cite{ArT08} that the block error probability for polar
coding and successive cancellation decoding is $O(2^{-2^{n\beta}})$ for any
fixed $\beta < \frac12$, where $2^n$ is the blocklength.
In this case we say that $G_2$ has {\em exponent} $\frac12$. We show that this
exponent can be improved by considering larger matrices. In fact, the exponent
can be made arbitrarily close to $1$ by increasing the size of the
matrix $G$.

Finally, we give an explicit construction of a family of matrices, derived from
BCH codes, with exponent approaching $1$ for large $\ell$. This
construction results in a matrix whose exponent
exceeds $\frac12$ for $\ell  = 16$. 
%Further, we show that this is the smallest
%$\ell$ at which an exponent larger than $\frac12$ can be achieved.

\section{Preliminaries}
In this paper we deal exclusively with \emph{symmetric} channels:
\begin{definition}
A binary-input discrete memoryless channel (B-DMC) $W: \{0,1\}\to \mathcal{Y}$ is said to be symmetric if there exists a permutation $\pi: \mathcal{Y}\to\mathcal{Y}$ such that $\pi = \pi^{-1}$ and $W(y|0) = W(\pi(y)|1)$ for all $y\in\mathcal{Y}$.
\end{definition}

Let $W: \{0,1\} \to \mathcal{Y}$ be a symmetric binary-input discrete memoryless channel (B-DMC). 
Let $I(W) \in [0,1]$ denote the mutual information between the input and output of $W$ 
with uniform distribution on the inputs.  Also, let $Z(W) \in [0,1]$ denote the Bhattacharyya 
parameter of $W$, i.e., $Z(W) = \sum_{y\in\mathcal{Y}} \sqrt{W(y|0)W(y|1)}$. 

Fix an $\ell\geq 3$ and an $\ell\times\ell$ invertible matrix $G$ with entries in $\{0,1\}$. Consider a random $\ell$-vector $U_1^\ell$ that is uniformly distributed over $\{0,1\}^\ell$. Let $X_1^\ell = U_1^\ell G$, where the multiplication is performed over GF(2). Also, let $Y_1^\ell$ be the output of $\ell$ uses of $W$ with the input $X_1^\ell$. The channel between $U_1^\ell$ and $Y_1^\ell$ is defined by the transition probabilities 
\begin{align}\label{equ:channelwl}
W_\ell (y_1^\ell\mid u_1^\ell) \triangleq \prod_{i=1}^\ell W(y_i\mid x_i) = \prod_{i=1}^\ell W(y_i\mid (u_1^\ell G)_i).
\end{align}
Define $W^{(i)}: \{0,1\} \to \mathcal{Y}^\ell \times \{0,1\}^{i-1}$ as the channel with input $u_i$, output $(y_1^\ell,u_1^{i-1})$ and transition probabilities
\begin{align} \label{equ:transitionprobabilities}
W^{(i)}(y_1^\ell,u_1^{i-1}\mid u_i) = \frac{1}{2^{\ell-1}} \sum_{u_{i+1}^\ell} W_\ell(y_1^\ell\mid u_1^\ell),
\end{align}
and let $Z^{(i)}$ denote its Bhattacharyya parameter, i.e., 
\begin{align*}
Z^{(i)} = \sum_{y_1^\ell,u_1^{i-1}} \sqrt{W^{(i)}(y_1^\ell,u_1^{i-1}\mid 0)W^{(i)}(y_1^\ell,u_1^{i-1}\mid 1)}.
\end{align*}

For $k\geq 1$ let $W^k:\{0,1\}\to\mathcal{Y}^k$ denote the B-DMC with transition probabilities 
\begin{align*}
W^k(y_1^k\mid x) = \prod_{j=1}^k W(y_j\mid x).
\end{align*}

Also let $\tilde{W}^{(i)}: \{0,1\} \to \mathcal{Y}^\ell$ denote the B-DMC with transition probabilities
\begin{align}\label{equ:channelwtilde} 
\tilde{W}^{(i)}(y_1^\ell\mid u_i) = \frac{1}{2^{\ell-i}} \sum_{u_{i+1}^\ell} W_\ell(y_1^\ell\mid 0_1^{i-1},u_i^\ell).
\end{align} 
\begin{observation}\label{obs:equivalent}
Since $W$ is symmetric, the channels $W^{(i)}$ and $\tilde{W}^{(i)}$ are equivalent in the sense that for any fixed $u_1^{i-1}$ there exists a permutation $\pi_{u_1^{i-1}}: \mathcal{Y}^\ell \to \mathcal{Y}^\ell$ such that 
\[
W^{(i)}(y_1^\ell, u_1^{i-1} \mid u_i) = \frac{1}{2^{i-1}}\tilde{W}^{(i)}(\pi_{u_1^{i-1}}(y_1^\ell)\mid u_i).
\]
\end{observation}
Finally, let $I^{(i)}$ denote the mutual information between the input and output of channel $W^{(i)}$. Since $G$ is invertible, it is easy to check that 
\begin{align*}
\sum_{i=1}^\ell I^{(i)} = \ell I(W). 
\end{align*}

We will use $\code{}$ to denote a linear code and $\dmin(\code{})$ to denote its minimum distance.
We let $\langle g_1,\dots,g_k\rangle$ denote the linear code generated by the vectors $g_1,\dots,g_k$. We let $\dh(a,b)$ denote the Hamming distance between binary vectors $a$ and $b$. We
also let $\dh(a,\code{})$ denote the minimum distance between a vector $a$ and a
code $\code{}$, i.e., $\dh(a,\code{}) = \min_{c \in \code{}}\dh(a,c)$.

\section{Polarization}
We say that $G$ is a \emph{polarizing} matrix if there exists an $i\in\{1,\dotsc,\ell\}$ for which 
%for which $\tilde{W}^{(i)}$ is equivalent to $W^k$ for 
%some $k\geq 2$, in the sense that
\begin{align}\label{eqn:polarizing}
\tilde{W}^{(i)}(y_1^\ell\mid u_i) = Q(y_{A^c}) \prod_{j\in A}W(y_j \mid u_i) 
\end{align}
for some and $A\subseteq \{1,\dots,\ell\}$ with $\vert A \vert = k$, $k\geq 2$, and a probability distribution $Q:\mathcal{Y}^{|A^c|}\to [0,1]$. 

In words, a matrix $G$ is polarizing if there exists a bit which ``sees''
a channel whose $k$ outputs are equivalent to those of $k$ independent realizations of 
the underlying channel, whereas the remaining $\ell-k$ outputs are independent of the input to the channel. The reason to call such a $G$ ``polarizing'' 
is that, as we will see shortly, 
a repeated application of such a transformation polarizes the underlying channel.

Recall that by assumption $W$ is symmetric. Hence, by Observation 
\ref{obs:equivalent}, equation (\ref{eqn:polarizing}) implies  
\begin{align}
W^{(i)}(y_1^\ell,u_1^{i-1} \mid u_i) = \frac{Q(y_{A^c})}{2^{i-1}} \prod_{j\in A}W((\pi_{u_1^{i-1}}(y_1^\ell))_j\mid u_i),
\end{align}
an equivalence we will denote by $W^{(i)} \equiv W^k$. Note that $W^{(i)} \equiv W^k$ 
implies $I^{(i)} = I(W^k)$ and $Z^{(i)} = Z(W^k)$.

We start by claiming that any invertible $\{0,1\}$ matrix $G$ can be written as a (real)
sum $G=P+P'$, where $P$ is a permutation matrix, and $P'$ is a $\{0,1\}$
matrix. 
To see this, consider a bipartite graph on $2 \ell$ nodes. The $\ell$ left nodes
correspond to the rows of the matrix and the $\ell$ right nodes correspond to the
columns of the matrix. Connect left node $i$ to right node $j$ if $G_{ij}=1$. The
invertibility of $G$ implies that for every subset of rows ${\mathcal R}$ 
the number of columns which contain non-zero elements in these rows is at least
$|{\mathcal R}|$. 
By Hall's Theorem \cite[Theorem 16.4.]{BoM08}
this guarantees that there is a matching between the left and the right nodes of the graph and
this matching represents a permutation.
Therefore, for any invertible matrix $G$, there exists a
column permutation so that all diagonal elements of the permuted matrix are $1$.
Note that the transition probabilities defining $W^{(i)}$ are invariant (up to a
permutation of the outputs $y_1^\ell$) under column permutations on $G$. Therefore, for
the remainder of this section, and without loss of generality, we assume
that $G$ has $1$s on its diagonal.

The following lemma gives necessary and sufficient conditions for 
(\ref{eqn:polarizing}) to be satisfied.

\begin{lemma}
[Channel Transformation for Polarizing Matrices]
\label{lem:square}
Let $W$ be a symmetric B-DMC.
\begin{itemize} 
\item[(i)] If $G$ is not upper triangular, then there exists an $i$ for which $W^{(i)} \equiv W^k$ for some $k\geq 2$.
\item[(ii)] If $G$ is upper triangular, then $W^{(i)} \equiv W$ for all $1\leq i \leq \ell$.
\end{itemize} 
\end{lemma}
\begin{proof}
Let the number of 1s in the last row of $G$ be $k$. Clearly $W^{(\ell)}
\equiv W^k$.  If $k\geq2$ then $G$ is not upper triangular and the first
claim of the lemma holds. If $k=1$ then
\begin{align}\label{equ:fact}
G_{lk} = 0, \;\;\text{for all $1 \leq k<\ell$}. 
\end{align}
One can then write
\begin{align*}
&W^{(\ell-i)}(y_1^\ell, u_1^{\ell-i-1}\mid u_{\ell-i}) \\
&= \frac{1}{2^{\ell-1}} \sum_{u_{\ell-i+1}^\ell} W_\ell(y_1^\ell \mid u_1^\ell) \\
& = \frac{1}{2^{\ell-1}} \sum_{u_{\ell-i+1}^{\ell-1},u_\ell} \Pr[Y_1^{\ell-1}= y_1^{\ell-1}\mid U_1^{\ell} = u_1^{\ell}] \\
&\phantom{xxxxxxxxxxxxx}\cdot\Pr[Y_\ell = y_\ell\mid Y_1^{\ell-1} = y_1^{\ell-1},U_1^{\ell} = u_1^{\ell}]\\
& \stackrel{(\ref{equ:fact})}{=} \frac{1}{2^{\ell-1}} \sum_{u_{\ell-i+1}^{\ell-1},u_\ell}
W_{\ell-1}(y_1^{\ell-1}\mid u_1^{\ell-1}) \\ 
&\phantom{xxxxxxxxxxxxx}\cdot\Pr[Y_\ell = y_\ell\mid Y_1^{\ell-1} = y_1^{\ell-1},U_1^{\ell} = u_1^{\ell}]\\
& = \frac{1}{2^{\ell-1}} \sum_{u_{\ell-i+1}^{\ell-1}}W_{\ell-1}(y_1^{\ell-1}\mid
u_1^{\ell-1}) \\
&\phantom{xxxxxxxxxxxxx} \cdot\sum_{u_\ell}\Pr[Y_\ell = y_\ell\mid Y_1^{\ell-1} = y_1^{\ell-1},U_1^{\ell} = u_1^{\ell}]\\
& = \frac{1}{2^{\ell-1}} \big[W(y_\ell\mid 0) + W(y_\ell\mid 1)\big]
\sum_{u_{\ell-i+1}^{\ell-1}} W_{\ell-1}(y_1^{\ell-1}\mid u_1^{\ell-1}).
\end{align*}
Therefore, $Y_\ell$ is independent of the inputs to the
channels $W^{(\ell-i)}$ for $i=1,\dotsc,\ell-1$. 
This is equivalent to
saying that channels $W^{(1)},\dotsc,W^{(\ell-1)}$ are defined by the
matrix $G^{(\ell-1)}$, where we define $G^{(\ell-i)}$ as the $(\ell-i)\times (\ell-i)$ matrix obtained 
from $G$ by removing its last $i$ rows and columns. 
Applying the same argument to $G^{(\ell-1)}$
and repeating, we see that if $G$ is upper triangular, then we have
$W^{(i)} \equiv W$ for all $i$. On the other hand, if $G$ is not upper
triangular, then there exists an $i$ for which $G^{(\ell-i)}$
has at least two 1s in the last row. This in turn implies that $W^{(i)} \equiv W^k$ 
for some $k\geq 2$. 
\end{proof}

Consider the recursive channel combining operation given in \cite{Ari08}, using
a transformation $G$. Recall that $n$ recursions of this construction is
equivalent to applying the transformation $A_nG^{\otimes n}$ to $U_1^{\ell^n}$
where, $A_n:\{1,\dotsc,\ell^n\}\to\{1,\dotsc,\ell^n\}$ is a permutation defined
analogously to the bit-reversal operation in \cite{Ari08}.
\begin{theorem}
[Polarization of Symmetric B-DMCs]
 \label{thm:main}
Given a symmetric B-DMC $W$ and an $\ell \times \ell$ transformation $G$,
consider the channels $W^{(i)}, i=\{1,\dotsc,\ell^n\}$, defined by the transformation $A_nG^{\otimes n}$. 
\begin{itemize}
\item[(i)] If $G$ is polarizing, then for any $\delta > 0$
\begin{align}\label{eqn:fractionI}
\lim_{n \to \infty}\frac{\lvert\left\{i\in\{1,\dotsc,\ell^n\}: I(W^{(i)}) \in
(\delta, 1-\delta) \right\} \rvert}{\ell^n} = 0, \\
\lim_{n \to \infty}\frac{\lvert\left\{i\in\{1,\dotsc,\ell^n\}: Z(W^{(i)}) \in
(\delta, 1-\delta) \right\} \rvert}{\ell^n} = 0. \label{eqn:fractionZ}
\end{align}
\item[(ii)] If $G$ is not polarizing, then for all $n$ and $i\in\{1,\dotsc,\ell^n\}$
\begin{align*}
I(W^{(i)}) = I(W), \;\; Z(W^{(i)}) = Z(W).
\end{align*}
\end{itemize}
\end{theorem}
In \cite[Section 6]{Ari08}, Ar\i kan proves part (i) of Theorem \ref{thm:main} for $G=G_2$. His proof involves defining a random variable $W_n$ that is uniformly distributed over the set $\{W^{(i)}\}_{i=1}^{\ell^n}$ (where $\ell=2$ for the case $G=G_2$), which implies
\begin{align}
\Pr[I(W_n)\in (a,b)] & =\frac{\lvert\left\{i\in\{1,\dotsc,\ell^n\}: I(W^{(i)}) \in
(a, b) \right\} \rvert}{\ell^n},\label{eqn:I-count}\\
\Pr[Z(W_n)\in (a,b)] & =\frac{\lvert\left\{i\in\{1,\dotsc,\ell^n\}: Z(W^{(i)}) \in
(a, b) \right\} \rvert}{\ell^n}.\label{eqn:Z-count}
\end{align}
Following Ar\i kan, we define the random variable $W_n \in \{W^{(i)}\}_{i=1}^{\ell^n}$ 
for our purpose through a tree process $\{W_n; n\geq 0\}$ with
\begin{align*}
W_0 & = W, \\
W_{n+1} & = W_{n}^{(B_{n+1})},
\end{align*}
where $\{B_n; n\geq 1\}$ is a sequence of i.i.d.\ random variables defined on a
probability space $(\Omega, \mathcal{F}, \mu )$, and where $B_n$ is uniformly
distributed over the set $\{1,\dotsc,\ell\}$. Defining $\mathcal{F}_0 =
\{\emptyset, \Omega\}$ and $\mathcal{F}_n = \sigma (B_1,\dotsc, B_n)$ for $n
\geq 1$, we augment the above process by the processes $\{I_n; n\geq 0\} := \{I(W_n); n\geq 0\}$ and
$\{Z_n; n\geq 0\} := \{Z(W_n); n\geq 0\}$. It is easy to verify that these processes satisfy (\ref{eqn:I-count}) and (\ref{eqn:Z-count}).
%\begin{align*}
%\frac{\lvert\left\{i\in\{1,\dotsc,\ell^n\}: I(W^{(i)}) \in (a,b)
%\right\} \rvert}{\ell^n} = \Pr[I_n \in (a,b)].
%\end{align*}

\begin{observation}\label{obs:Imartingale}
$\{(I_n, \mathcal{F}_n)\}$ is a bounded martingale and therefore converges w.p.\ 1
and in $\mathcal{L}^1$ to a random variable $I_\infty$.
\end{observation}

\begin{lemma}[$I_\infty$] \label{lem:I-converges}
If $G$ is polarizing, then
\begin{align*}
I_\infty =
\begin{cases}
1 & \textrm{w.p. } I(W), \\
0 & \textrm{w.p. } 1 - I(W).
\end{cases}
\end{align*}
\end{lemma}
\begin{proof}
For any polarizing transformation $G$, Lemma~\ref{lem:square} implies that there
exists an $i\in\{1,\dotsc,\ell\}$ and $k\geq 2$ for which 
\begin{align}
I^{(i)} & = I(W^k). \label{eqn:prod-channel}
\end{align}
This implies that for the tree process defined above, we have
\begin{align*}
I_{n+1} = I(W_n^k) \textrm{ with probability at least } \frac{1}{\ell},
\end{align*}
for some $k\geq 2$. 
Moreover by the convergence in $\mathcal{L}^1$ of $I_n$, we have $\mathbb{E}
[|I_{n+1} - I_n | ] \stackrel{n \to \infty}{\longrightarrow}0$.
This in turn implies 
\begin{align}
\mathbb{E} [|I_{n+1} - I_n |] \geq \frac{1}{\ell} \mathbb{E}[(I(W_n^k) - I(W_n)] \to 0. \label{eqn:I-repetition}
\end{align}
It is shown in Lemma~\ref{lem:combinedI} in the Appendix that for any symmetric B-DMC $W_n$, if $I(W_n) \in (\delta,1-\delta)$ for some $\delta > 0$, then there exists an $\eta (\delta) > 0$ such that $I(W_n^k) - I(W_n) > \eta (\delta)$. Therefore, convergence in (\ref{eqn:I-repetition}) implies $I_\infty \in \{0,1\}$ w.p.\ 1. The claim on the probability distribution of $I_\infty$ follows from the fact that $\{I_n\}$ is a martingale, i.e., $\mathbb{E} [I_\infty] = \mathbb{E} [I_0] = I(W)$.
\end{proof}

{\em Proof of Theorem \ref{thm:main}: }
Note that for any $n$ the fraction in
\eqref{eqn:fractionI} is equal to $\Pr[I_n \in (\delta,
1-\delta)]$. Combined with Lemma~\ref{lem:I-converges}, this implies
\eqref{eqn:fractionI}. 

For any B-DMC $Q$, $I(Q)$ and $Z(Q)$ satisfy 
\cite{Ari08}
\begin{align*}
I(Q)^2 + Z(Q)^2 \leq 1, \\
I(Q) + Z(Q) \geq 1.
\end{align*}
When $I(Q)$ takes on the value $0$ or $1$, these two inequalities imply that
$Z(Q)$ takes on the value $1$ or $0$, respectively. 
From Lemma~\ref{lem:I-converges} we know that $\{I_n\}$ converges to  $I_\infty$ w.p. $1$
and $I_\infty \in \{0,1\}$. This
implies that $\{Z_n\}$ converges w.p. $1$ to a random variable $Z_\infty$ and
\begin{align*}
Z_\infty =
\begin{cases}
0 & \textrm{w.p. } I(W), \\
1 & \textrm{w.p. } 1 - I(W).
\end{cases}
\end{align*}
This proves the first part of the theorem.
The second part follows from 
Lemma~\ref{lem:square}, (ii). \qed

\begin{remark}
Ar\i kan's proof for part (i) of Theorem \ref{thm:main} with $G=G_2$ proceeds by first showing the
convergence of $\{Z_n\}$, instead of $\{I_n\}$. This is accomplished by
showing that for the matrix $G_2$ the resulting process $\{Z_n\}$ is a
submartingale. Such a property is in general difficult to prove for arbitrary $G$. On
the other hand, the process $\{I_n\}$ is a martingale for any invertible matrix
$G$, which is sufficient to ensure convergence.
\end{remark}

Theorem~\ref{thm:main} guarantees that repeated application of a polarizing 
matrix $G$ polarizes the underlying channel $W$, i.e.,
the resulting channels $W^{(i)}$, $i\in \{1,\dots,\ell^n\}$, tend towards
either a noiseless or a completely noisy channel. Lemma~\ref{lem:I-converges}
ensures that the fraction of noiseless channels
is indeed $I(W)$. This suggests to use the noiseless channels for transmitting
information while transmitting no information over the noisy channels \cite{Ari08}. 
Let $\mathcal{A}\subset \{1,\dotsc,\ell^n\}$ denote the set of channels $W^{(i)}$ used for transmitting the information bits. 
%as the set (or one of the many sets) of channels with the smallest $2^{\ell^{n}R}$ $Z^{(i)}$s, this rule at finite blocklengths translates to transmitting information over the channels in $\mathcal{A}(R)$, at rate $R$.
Since $Z^{(i)}$ upper bounds the error probability of decoding bit $U_i$ with the knowledge of $U_1^{i-1}$, the block error probability of such a transmission scheme under successive cancellation decoder can be upper bounded as \cite{Ari08}
\begin{align}\label{eqn:upperbnd}
P_B \leq \sum_{i\in\mathcal{A}} Z^{(i)}.
\end{align}
Further, the block error probability can also be lower bounded in terms of the $Z^{(i)}$s: Consider a symmetric	B-DMC with Bhattacharyya parameter $Z$, and let $P_e$ denote the bit error probability of uncoded transmission over this channel. It is known that 
\[
P_e \geq \frac 12 ( 1-\sqrt{1-Z^2}).
\]
A proof of this fact is provided in the Appendix. Under successive cancellation decoding,
 the block error probability is lower bounded by each of the bit error probabilities over the channels $W^{(i)}$. Therefore the former quantity can be lower bounded by 
\begin{align}\label{eqn:lowerbnd}
P_B \geq \max_{i\in\mathcal{A}} \frac 12 (1-\sqrt{1-(Z^{(i)})^2}).
\end{align} 

Both the above upper and lower bounds to the block error probability look somewhat loose at a first look. However, as we shall see later, these bounds are sufficiently tight for our purposes. Therefore, it suffices to analyze the behavior of the $Z^{(i)}$s.

%By using the union bound as done in \cite{Ari08}, Theorem~\ref{thm:universalBound}
%in turn implies that the 
%block error probability for polar codes under successive cancellation decoding decays at least like
%$2^{-N^\beta}$ for any $\beta < \frac{\log_\ell 2}{\ell}$.
%The above definition implies that the block error probability of polar codes under successive cancellation decoding is upper bounded by $2^{-\ell^{n\beta}}$ for any $\beta < \exponent(G)$ \cite{Ari08}.

\section{Rate of Polarization}
For the matrix $G_2$ 
Ar\i kan shows that, combined with successive cancellation decoding,
these codes achieve a vanishing block error probability for any rate strictly
less than $I(W)$. Moreover, it is shown in \cite{ArT08} that when $Z_n$ 
approaches 0 it does so at a sufficiently fast rate:
\begin{theorem}[\cite{ArT08}]
Given a B-DMC $W$, the matrix $G_2$ and any
$\beta < \frac12$,
\begin{align*}
\lim_{n \to \infty} \Pr[Z_n \leq 2^{-2^{n\beta}}] = I(W).
\end{align*}
\end{theorem}

A similar result for arbitrary $G$ is given in the following theorem. 
%This in turn implies that the block 
%error probability of polar codes under successive cancellation decoding is 
%$o(2^{-\ell^{n\beta}})$ for any $\beta < \frac{\log_\ell2}{\ell}$.

\begin{theorem}[Universal Bound on Rate of Polarization]\label{thm:universalBound}
Given a symmetric B-DMC $W$, an $\ell\times\ell$ polarizing matrix $G$, and any
$\beta < \frac{\log_\ell 2}{\ell}$,
\begin{align*}
\lim_{n \to \infty} \Pr[Z_n \leq 2^{-\ell^{n\beta}}] = I(W).
\end{align*}
\end{theorem}
{\em Proof Idea: }
For any polarizing matrix it can be shown that $Z_{n+1}\leq \ell Z_n$
with probability 1 and that $Z_{n+1} \leq Z_n^2$  with probability at least
$1/\ell$. The proof then follows by adapting the proof of \cite[Theorem 3]{ArT08}.
\qed

%We have seen that any $\ell \times \ell$ matrix none of whose column
%permutations is upper triangular polarizes symmetric B-DMCs. Moreover, combining
%this construction with successive cancellation decoding, a block
%error probability of $2^{-{\ell}^{n\beta}}$ for any $\beta < \frac{\log_\ell
%2}{\ell}$ can be achieved. 
The above estimation of the probability is universal and is independent of
the exact structure of $G$. We are now interested in a more precise
estimate of this probability.
The results in this section are the natural generalization of those in \cite{ArT08}.
%, where
%it was shown that under successive cancellation decoding 
%the matrix $G_2$ achieves a block error probability of
%$o(2^{-2^{n\beta}})$ for any $\beta < \frac12$.

\begin{definition}[Rate of Polarization]\label{def:rate}
For any B-DMC $W$ with $0<I(W)<1$, 
we will say that an $\ell\times\ell$ matrix $G$ has rate of polarization $\exponent(G)$ if
\begin{itemize}
\item[(i)] For any fixed $\beta < \exponent(G)$,
%\begin{align}\label{def:rate1} 
\[
\liminf_{n \to \infty} \Pr[Z_n \leq 2^{-\ell^{n\beta}}] = I(W).
\]
%\end{align}
\item[(ii)] For any fixed $\beta > \exponent(G)$,
%\begin{align} \label{def:rate2}
\[
\liminf_{n\to\infty} \Pr[Z_n \geq 2^{-\ell^{n\beta}}] = 1.
\]
%\end{align}
\end{itemize}
\end{definition}
For convenience, in the rest of the paper we refer to $\exponent(G)$ as the exponent of the
matrix $G$.

The definition of exponent provides a meaningful performance measure of polar codes under successive cancellation decoding. This can be seen as follows: Consider a matrix $G$ with exponent $\exponent(G)$. Fix $0<R<I(W)$ and $\beta < \exponent(G)$. Definition \ref{def:rate} (i) implies that for $n$ sufficiently large there exists a set $\mathcal{A}$ of size $\ell^{n}R$ such that $\sum_{i\in \mathcal{A}} Z^{(i)} \leq 2^{-\ell^{n\beta}}$. Using set $\mathcal{A}$ as the set of information bits, the block error probability under successive cancellation decoding $P_B$ can be bounded using
\eqref{eqn:upperbnd} as 
\begin{align*}
P_B \leq 2^{-\ell^{n\beta}}.
\end{align*} 
Conversely, consider $R>0$ and $\beta> \exponent(G)$. Definition \ref{def:rate} (ii) implies that for $n$ sufficiently large, any set $\mathcal{A}$ of size $\ell^nR$ will satisfy $\max_{i\in\mathcal{A}} Z^{(i)} > 2^{-\ell^{n\beta}}$. Using \eqref{eqn:lowerbnd} the block
error probability can be lower bounded as 
\begin{align*}
P_B \geq 2^{-\ell^{n\beta}}.
\end{align*}

It turns out, and it will be shown later, that the exponent is independent of the
channel $W$. Indeed, we will show in Theorem~\ref{thm:l-ary} that the exponent
$\exponent(G)$ can be expressed as a
function of the {\em partial distances} of $G$.
\begin{definition}[Partial Distances]\label{def:partial-d}
Given an $\ell\times\ell$ matrix $G = [g_1^T,\dotsc,g_\ell^T]^T$, we define the \emph{partial distances} $D_i$, $i=1,\dotsc,\ell$ as
\begin{align*}
D_i & \triangleq \dh(g_i,\langle g_{i+1},\dotsc,g_\ell\rangle), \qquad i=1,\dotsc,\ell-1, \\
D_\ell & \triangleq \dh(g_\ell,0).
\end{align*}
\end{definition}
\begin{example}\label{ex:partial-d}
The partial distances of the matrix 
\begin{align*}
F = \left[
\begin{array}{ccc}
1& 0& 0 \\
1& 0& 1 \\
1& 1& 1 \\
\end{array}\right]
\end{align*}
are $D_1 = 1, D_2 = 1,D_3 =3$.
\end{example}
In order to establish the relationship between $\exponent(G)$ and the partial distances of $G$
we consider the Bhattacharyya parameters $Z^{(i)}$ of the channels $W^{(i)}$. These parameters 
depend on $G$ as well as on $W$. The exact
relationship with respect to $W$ is difficult to compute in general. 
However, there are sufficiently tight upper and lower bounds
on the $Z^{(i)}$s in terms of $Z(W)$, the Battacharyya parameter of $W$.
\begin{lemma}[Bhattacharyya Parameter and Partial Distance]\label{lem:z-bound}
For any symmetric B-DMC $W$ and any $\ell\times\ell$ matrix $G$ with partial
distances $\{D_i\}_{i=1}^\ell$
\begin{align}\label{eqn:z-bound}
Z(W)^{D_i} \leq Z^{(i)} \leq 2^{\ell-i} Z(W)^{D_i}.
\end{align} 
\end{lemma}
\begin{proof}
To prove the upper bound we write
\begin{align}
Z^{(i)} &= \sum_{y_1^\ell,u_1^{i-1}}\sqrt{W^{(i)}(y_1^\ell, u_1^{i-1}\mid 0) W^{(i)}(y_1^\ell, u_1^{i-1}\mid 1)} \notag\\
& \stackrel{(\ref{equ:transitionprobabilities})}{=} \frac{1}{2^{\ell-1}}\sum_{y_1^\ell,u_1^{i-1}} \notag\\
&\phantom{xx}\sqrt{\sum_{v_{i+1}^{\ell},w_{i+1}^\ell} W_\ell(y_1^\ell\mid u_1^{i-1}, 0, v_{i+1}^\ell) W_\ell(y_1^\ell\mid u_1^{i-1}, 1, w_{i+1}^\ell)} \notag\\ 
&\leq \frac{1}{2^{\ell-1}}\sum_{y_1^\ell,
u_1^{i-1}}\sum_{v_{i+1}^\ell,w_{i+1}^\ell} \sqrt{W_\ell(y_1^\ell\mid u_1^{i-1},
0, v_{i+1}^\ell)}\notag \\
 &\phantom{xxxxxxxxxxxxxxxxxxx}\cdot\sqrt{W_\ell(y_1^\ell\mid u_1^{i-1}, 1, w_{i+1}^\ell)}. \label{eqn:root}
\end{align} Let $c_0 = (u_1^{i-1},0,v_{i+1}^\ell) G$ and $c_1 = (u_1^{i-1},1,w_{i+1}^\ell) G$. 
Let $S_0(S_1)$ be the set of indices where both $c_0$ and $c_1$ are equal to $0(1)$. Let $S^c$ be the complement of $S_0\cup S_1$. We have 
\begin{align*}
|S^c| = \dh(c_0, c_1) \geq D_i.
\end{align*}
Now, (\ref{eqn:root}) can be rewritten as 
\begin{align*}
Z^{(i)}&\leq
\frac{1}{2^{\ell-1}}\sum_{v_{i+1}^{\ell},w_{i+1}^\ell}\sum_{y_1^\ell,u_1^{i-1}} \prod_{j\in S_0} W(y_j\mid 0) \prod_{j\in S_1} W(y_j\mid 1)\\
&\phantom{xxxxxx}\cdot\sqrt{\prod_{j\in S^c}W(y_j\mid 0) W(y_j\mid 1)} \\
&\leq \frac{1}{2^{\ell-1}}\sum_{v_{i+1}^{\ell},w_{i+1}^\ell,u_1^{i-1}} Z^{D_i}\\
&= 2^{\ell-i} Z^{D_i}.
\end{align*}

For the lower bound on $Z^{(i)}$, first note that by Observation
\ref{obs:equivalent}, we have $Z(W^{(i)}) = Z(\tilde{W}^{(i)})$. Therefore it
suffices to show the claim for the channel $\tilde{W}^{(i)}$. Let $G = [g_1^T,\dots,g_\ell^T]^T$. Then using \eqref{equ:channelwl}, \eqref{equ:transitionprobabilities} and
\eqref{equ:channelwtilde}, $\tilde{W}^{(i)}$ can be written as
\begin{align}
\tilde{W}^{(i)} (y_1^\ell\mid u_i) =
\frac{1}{2^{\ell-i}}\sum_{x_1^\ell \in \mathcal{A}(u_i)}\prod_{k=1}^\ell W(y_k\vert x_k)
\end{align}
where $x_1^\ell \in \mathcal{A}(u_i)\subset \{0,1\}^\ell$ if and only if for some $u_{i+1}^\ell \in \{0,1\}^{\ell-i}$ 
%\begin{align} 
%\mathcal{A}(u_i) = \{x_1^\ell: x_1^\ell = (0_1^{i-1},u_i,u_{i+1}^\ell)G \textrm{ for some } u_{i+1}^\ell\}.
%\end{align} 

\begin{align}\label{equ:wtildewithG}
x_1^\ell = u_i g_i + \sum_{j=i+1}^\ell u_jg_j.
\end{align}
Consider the code $\langle g_{i+1},\dots,g_\ell\rangle$ and
let $\sum_{j=i+1}^\ell \alpha_j g_j$ be a codeword satisfying
$\dh(g_i,\sum_{j=i+1}^\ell \alpha_j g_j) = D_i$. Due to the linearity of the code $\langle g_{i+1}\dotsc,g_\ell\rangle$, 
one can equivalently say that $ x_1^\ell \in \mathcal{A}(u_i)$ if and only if 
\begin{align}\label{eqn:wtildewithGprime}
x_1^\ell = u_i\big( g_i+
\sum_{j=i+1}^\ell \alpha_j g_j\big) +
\sum_{j=i+1}^\ell u_j g_j.
\end{align}

Now let $g'_i = g_i + \sum_{j=i+1}^\ell \alpha_j g_j$ and $G' =
[g_1^T,\dots,g_{i-1}^T,{g'}^T_i,g_{i+1}^T,\dots,g_\ell^T]^T$. Equations 
(\ref{equ:wtildewithG}) and (\ref{eqn:wtildewithGprime}) show 
that the channels $W^{(i)}$ defined by the matrices  $G$ and $G'$ are equivalent. 
Note that $G'$ has the property that the Hamming
weight of $g'_i$ is equal to  $D_i$. 
%Hence, without loss of generality, we
%can assume in the sequel that $G$ has the property that all its rows 
%have weights equal to their corresponding partial distances.

We will now consider a channel $W_{g}^{(i)}$ where a genie provides extra information to the decoder. Since $\tilde{W}^{(i)}$ is degraded with respect to the
genie-aided channel $W_{g}^{(i)}$, and since the ordering of the
Bhattacharyya parameter is preserved under degradation, it suffices to find a
genie-aided channel for which $Z_{g}^{(i)} = Z(W)^{D_i}$.

Consider a genie which reveals the bits $u_{i+1}^\ell$ to the decoder (Figure \ref{fig:geniechannel}).
With the knowledge of $u_{i+1}^\ell$ the decoder's task reduces to finding the
value of any of the transmitted bits $x_j$ for which $g_{ij} = 1$. Since each bit $x_j$ goes
through an independent copy of $W$, and since the weight of $g_i$ is equal to $D_i$,
the resulting channel $W_g^{(i)}$ is equivalent to $D_i$ independent copies of
$W$. Hence, $Z_g^{(i)} = Z(W)^{D_i}$.

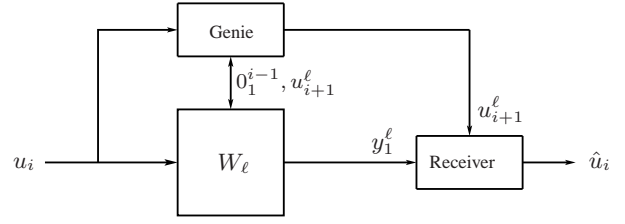
\begin{figure}
\begin{centering}
\input{ps/geniechannel}
\caption{Genie-aided channel $W^{(i)}_g$.}\label{fig:geniechannel}
\end{centering}
\end{figure}
\end{proof}
Lemma~\ref{lem:z-bound} shows that the link between 
$Z^{(i)}$ and $Z(W)$ is given in terms of the partial distances of $G$. 
This link is sufficiently strong to completely characterize $\exponent(G)$.
\begin{theorem}[Exponent from Partial Distances]\label{thm:l-ary}
For any symmetric B-DMC $W$ and any $\ell\times\ell$ matrix $G$ with partial distances $\{D_i\}_{i=1}^\ell$, the rate
of polarization $\exponent(G)$ is given by
\begin{align}
\exponent(G) = \frac{1}{\ell}\sum_{i=1}^\ell \log_\ell D_i.\label{eqn:rate-distance}
\end{align}
\end{theorem}
\begin{proof}
The proof is similar to that of \cite[Theorem 3]{ArT08}. We highlight the main idea and omit the details.

First note that by Lemma~\ref{lem:z-bound} we have $Z_j \geq Z_{j-1}^{D_{B_j}}$. Let $m_i= \lvert \{1\leq j \leq n: B_j= i\} \rvert$. We then obtain
\begin{align}\label{equ:exponentlowerbound}
Z_n \geq Z^{{\prod_i} D_i^{m_i}} = Z^{\ell^{(\sum_i m_i \log_\ell D_i)}}.
\end{align}
The exponent of $Z$ on the right-hand side of (\ref{equ:exponentlowerbound}) can be rewritten as
\begin{align*}
\ell^{\sum_i m_i \log_\ell D_i} = (\ell^n)^{\sum_i \frac{m_i}{n} \log_\ell D_i}.
\end{align*}
By the law of large numbers, for any $\epsilon > 0$,  
\begin{align*}
\left\vert \frac{m_i}{n}- \frac1\ell\right\vert \leq \epsilon \;\;\; 
\end{align*}
with high probability for $n$ sufficiently large. 
%which implies that the exponent of $Z$ goes to $\frac{1}\ell \sum_i\log_\ell D_i$ with high probability. 
This proves part $(ii)$ of the definition of $\exponent(G)$, i.e., for any
$\beta > \frac{1}{\ell} \sum_i \log_\ell D_i$,
\begin{align*}
\lim_{n\to\infty}\Pr[Z_n \geq 2^{-\ell^{n\beta}}] = 1.
\end{align*} 
The proof for part $(i)$ of the definition follows using similar arguments as above, and by noting that $Z_j\leq 2^{\ell - B_j} Z_{j-1}^{D_{B_j}}$. The constant $2^{\ell - B_j}$ can be taken care of using the `bootstrapping' argument of \cite{ArT08}.
%for any $\beta < \frac{1}{\ell} \sum_i \log_\ell D_i$,
%\begin{align*}
%\lim_{n\to\infty}\Pr[Z_n \leq 2^{-N^\beta}] = I(W).
%\end{align*} 
\end{proof}

\begin{example}
For the matrix $F$ considered in Example \ref{ex:partial-d}, we have
\begin{align*}
\exponent(F) = \frac13 (\log_3 1 + \log_3 1 + \log_3 3) = \frac13.	
\end{align*}
\end{example}

\section{Bounds on the Exponent}
For the matrix $G_2$, we have $\exponent(G_2) =
\frac12$. Note that for the case of $2\times 2$ matrices, the only polarizing
matrix is $G_2$. In order to address the question of
whether the rate of polarization can be improved by considering large matrices,
we define
\begin{align}
\exponent_{\ell} \triangleq \max_{G\in \{0,1\}^{\ell\times\ell}} \exponent(G).\label{eqn:max-exp}
\end{align}
Theorem \ref{thm:l-ary} facilitates the computation of $\exponent_\ell$ 
by providing an expression for $\exponent(G)$ in terms of the partial distances of $G$. Lemmas \ref{lem:mindist} and \ref{lem:monotonic} below provide further simplification for computing (\ref{eqn:max-exp}).

\begin{lemma}[Gilbert-Varshamov Inequality for Linear Codes]\label{lem:mindist}
Let $\code{}$ be a binary linear code of length $\ell$ and $\dmin(\code{}) = d_1$. Let
$g \in \{0,1\}^\ell$ and let $\dh(g,\code{}) = d_2$. Let
$\code{'}$ be the linear code obtained by adding the vector $g$ to $\code{}$,
i.e., $\code{'} = \langle g, \code{}\rangle$. Then $\dmin(\code{'}) = \min\{d_1,d_2\}$.
\end{lemma}
\begin{proof}
Since $\code{'}$ is a linear code, its codewords are of the form $c+\alpha g$ where $c\in \code{}, \alpha \in \{0,1\}$. Therefore
\begin{align*}
\dmin(\code{'}) 
&= \min_{c\in\code{}}\{\min\{\dh(0,c),\dh(0,c+g)\}\}\\
& =\min\{\min_{c\in\code{}}\{\dh(0,c)\},\min_{c\in\code{}}\{\dh(g,c)\}\}\\
&=\min\{d_1,d_2\}.
\end{align*}
\end{proof}

\begin{corollary}\label{cor:mindist}
Given a set of vectors $g_1,\dots,g_k$ with partial distances $D_j =
\dh(g_j,\langle g_{j+1},\dots,g_k\rangle)$, $j=1,\dotsc,k$,  
the minimum distance of the linear code $\langle g_1,\dots, g_k\rangle$ is given by
$\min_{j=1}^\ell\{D_j\}$.
\end{corollary}

%The problem of finding the best exponent $\exponent_\ell$ for an $\ell
%\times \ell$ matrix can be framed as a maximization problem \eqref{eqn:max-exp}.
The maximization problem in \eqref{eqn:max-exp} is not feasible in practice even for $\ell \geq 10$.
The following lemma allows to restrict this maximization to a smaller set of 
matrices. Even though the maximization problem still remains intractable, by
working on this restricted set, we obtain lower and upper bounds on $\exponent_\ell$.
\begin{lemma}[Partial Distances Should Decrease]\label{lem:monotonic}
Let $G = [g_1^T \dots g_\ell^T]^T$. Fix $k \in \{1,\dotsc,\ell\}$ and let $G' = [g_1^T \dots g_{k+1}^T g_k^T \dots g_\ell^T]^T$ be the matrix obtained from $G$ by swapping $g_k$ and $g_{k+1}$. Let $\{D_i\}_{i=1}^\ell$ and $\{D'_i\}_{i=1}^\ell$ denote the partial distances of $G$ and $G'$ respectively. If $D_k > D_{k+1}$, then
\begin{enumerate}
\item[(i)]  $\exponent(G') \geq \exponent(G)$,
\item[(ii)] $D'_{k+1} > D'_k$.
\end{enumerate}
\end{lemma}
\begin{proof}
Note first that $D_i = D'_i$ if $i\notin \{k,k+1\}$. Therefore, to prove the first claim, 
it suffices to show that $D'_kD'_{k+1} \geq D_kD_{k+1}$. To that end, write
\begin{align*}
D'_k &= \dh(g_{k+1},\langle g_k,g_{k+2},\dots,g_{\ell}\rangle),\\
D_k &= \dh(g_k,\langle g_{k+1},\dots,g_{\ell}\rangle),\\
D'_{k+1} &= \dh(g_{k},\langle g_{k+2},\dots,g_{\ell}\rangle),\\
D_{k+1} &= \dh(g_{k+1},\langle g_{k+2},\dots,g_{\ell}\rangle),
\end{align*}
and observe that $D'_{k+1} \geq D_k$ since $\langle g_{k+2},\dots,g_\ell\rangle$ is a sub-code of $\langle g_{k+1},\dots,g_\ell\rangle$. $D'_k$ can be computed as
\begin{align*}
%D'_k & =
&\min\left\{\min_{c\in \langle g_{k+2},\dots,g_\ell\rangle}\dh(g_{k+1},c), \min_{c\in \langle g_{k+2},\dots,g_\ell\rangle} \dh(g_{k+1},c+g_k)\right\}\\
&= \min\{D_{k+1}, \min_{c\in \langle g_{k+2},\dots,g_\ell\rangle}
\dh(g_{k},c+g_{k+1})\}\\
&= D_{k+1},
\end{align*}
where the last equality follows from 
\begin{align*}
\min_{c\in \langle g_{k+2},\dots,g_\ell\rangle}\dh(g_{k},c+g_{k+1})&\geq
\min_{c\in \langle g_{k+1},g_{k+2},\dots,g_\ell\rangle}\dh(g_{k},c) \\
&= D_k > D_{k+1}.
\end{align*}
Therefore, $D'_{k}D'_{k+1} \geq D_k D_{k+1}$, which proves the first claim. The second claim follows from  the inequality $D'_{k+1} \geq D_k > D_{k+1} = D'_k$.
\end{proof}

\begin{corollary}\label{cor:monotonic}
In the definition of $\exponent_\ell$ \eqref{eqn:max-exp}, the maximization can be restricted to
the matrices $G$ which satisfy $D_1\leq D_2 \leq \dotsc \leq D_\ell$.
\end{corollary}

\subsection{Lower Bound}
The following lemma provides a lower bound on $\exponent_\ell$ by using a 
Gilbert-Varshamov type construction.
\begin{lemma}[Gilbert-Varshamov Bound]\label{lem:lowerbound}
\begin{align*}
\exponent_\ell \geq \frac{1}{\ell} \sum_{i=1}^\ell \log_\ell \tilde{D}_i
\end{align*}
where 
\begin{align}\label{eqn:lower-dist}
\tilde{D}_i = \max \left\{D: \sum_{j=0}^{D-1} \binom{\ell}{j} < 2^i\right\}.
\end{align}
%\begin{align*}
%\exponent_\ell \geq \max_{\{D_i\}: \sum_{j=0}^{D_i-1} {\ell \choose j} < 2^{i} }\frac1\ell \sum_{j=1}^\ell \log _\ell D_j.
%\end{align*}
\end{lemma}
\begin{proof}
We will construct a matrix $G = [g_1^T,\dotsc,g_\ell^T]^T$, 
with partial distances $D_i = \tilde{D}_i$: Let $S(c,d)$ denote the set of binary vectors with Hamming distance at most $d$ from $c\in\{0,1\}^\ell$, i.e., 
\begin{align*}
S(c,d)= \{x\in \{0,1\}^{\ell} : \dh(x,c) \leq d\}.
\end{align*}
To construct the $i^{th}$ row of $G$ with partial distance $\tilde{D}_i$, we will find a $v\in\{0,1\}^\ell$ satisfying $\dh(v,\langle g_{i+1},\dotsc,g_\ell\rangle)=\tilde{D}_i$ and set $g_i=v$. Such a $v$ satisfies $v\notin S(c,\tilde{D}_i-1)$ for all $c \in \langle g_{i+1},\dotsc,g_\ell\rangle$ and exists if the sets $S(c,\tilde{D}_i-1)$, $c\in \langle g_{i+1},\dotsc,g_\ell\rangle$ do not cover $\{0,1\}^\ell$. The latter condition is satisfied if
\begin{align*}
\vert \cup_{c \in \langle g_{i+1},\dotsc,g_\ell\rangle} S(c,\tilde{D}_i-1)\vert  & \leq \sum_{c\in \langle g_{i+1},\dotsc,g_\ell\rangle}
\vert S(c,\tilde{D}_i -1)\vert \\
& = 2^{\ell-i} \sum_{j=0}^{\tilde{D}_i-1} \binom{\ell}{j} < 2^\ell,
\end{align*}
which is guaranteed by (\ref{eqn:lower-dist}).
\end{proof}
\begin{figure}[htp]
\centering
\input{ps/bounds}
\caption{The solid curve shows the lower bound on $\exponent_\ell$
as described by Lemma~\ref{lem:lowerbound}. The dashed curve corresponds to
the upper bound on $\exponent_\ell$ according to Lemma~\ref{lem:strongerupperbound}.
The points show the performance of the best matrices 
obtained by the procedure described in
Section~\ref{sec:construction}.}\label{fig:bound}
\end{figure}
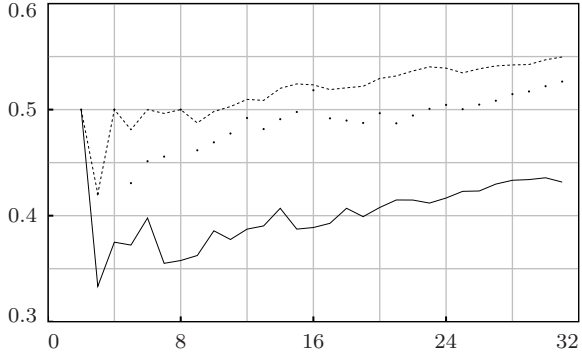
The solid line in Figure \ref{fig:bound} shows the lower bound of Lemma \ref{lem:lowerbound} . The
bound exceeds $\frac12$ for $\ell = 85$, suggesting that the exponent can be
improved by considering large matrices. In fact, the lower bound tends to
$1$ when $\ell$ tends to infinity:
\begin{lemma}[Exponent $1$ is Achievable] \label{lem:exponent1}
$\lim_{\ell \to \infty} \exponent_\ell = 1$.
\end{lemma}
\begin{proof}
Fix $\alpha\in(0,\frac 12)$. Let $\{\tilde{D}_j\}$ be defined as in Lemma
\ref{lem:lowerbound}. It is known (cite something here) that
$\tilde{D}_{\lceil\alpha\ell \rceil}$ in (\ref{eqn:lower-dist}) satisfies
$\lim_{\ell\to\infty} \tilde D_{\lceil\alpha \ell\rceil} \geq \ell h^{-1}(\alpha)$,
where $h(\cdot)$ is the binary entropy function. Therefore, there exists an
$\ell_0(\alpha)<\infty$ such that for all $\ell \geq \ell_0(\alpha)$ we have $D_{\lceil\alpha \ell\rceil} \geq \frac 12\ell h^{-1}(\alpha)$. Hence, for $\ell\geq\ell_0(\alpha)$ we can write
\begin{align*}
\exponent_\ell &\geq \frac{1}{\ell}\sum_{i=\lceil\alpha \ell\rceil}^\ell \log_{\ell}\tilde{D}_i\\
&\geq \frac{1}{\ell}(1-\alpha) \ell \log_{\ell}\tilde{D}_{\lceil\alpha\ell\rceil}\\
&\geq \frac{1}{\ell}(1-\alpha) \ell \log_{\ell}\frac{\ell h^{-1}(\alpha)}{2}\\
& = 1-\alpha + (1-\alpha)\log_{\ell}\frac{h^{-1}(\alpha)}{2},
\end{align*}
where the first inequality follows from Lemma \ref{lem:lowerbound}, and the second inequality follows from the fact that $\tilde{D}_i \leq \tilde{D}_{i+1}$ for all $i$. Therefore we obtain
\begin{align}\label{eqn:liminf}
\liminf_{\ell\to\infty} \exponent_\ell \geq 1-\alpha \quad \forall \alpha\in (0,\frac 12).
\end{align}
Also, since $\tilde{D}_i \leq \ell$ for all $i$, we have $\exponent_\ell \leq 1$ for all $\ell$. Hence,
\begin{align}\label{eqn:limsup}
\limsup_{\ell\to\infty} \exponent_\ell \leq 1.
\end{align}
Combining (\ref{eqn:liminf}) and (\ref{eqn:limsup}) concludes the proof.
\end{proof}

\subsection{Upper Bound}\label{sub:upperbound}
Corollary~\ref{cor:monotonic} says that for any $\ell$, there exists a matrix with
$D_1\leq\dots\leq D_\ell$ that achieves the exponent $\exponent_\ell$.
Therefore, to obtain upper bounds on $\exponent_\ell$, it suffices to bound
the exponent achievable by this restricted class of matrices. The partial
distances of these matrices can be bounded easily as shown in the following
lemma.
\begin{lemma}[Upper Bound on Exponent]\label{lem:generalbound}
Let $d(n,k)$ denote the largest possible minimum distance of a binary code of
length $n$ and dimension $k$. Then, 
\begin{align*}
\exponent_\ell \leq \frac 1\ell \sum_{i=1}^\ell \log_{\ell}d(\ell,\ell - i+1).
\end{align*}
\end{lemma}
\begin{proof}
Let $G$ be an $\ell\times\ell$ matrix with partial distances $\{D_i\}_{i=1}^\ell$ such that $\exponent(G) = \exponent_\ell$. Corollary \ref{cor:monotonic} lets us assume without loss of generality that $D_i\leq D_{i+1}$ for all $i$. We therefore obtain 
\[
D_i = \min_{j\geq i} D_j = \dmin(\langle g_i,\dotsc,g_\ell\rangle)\leq d(\ell,\ell-i+1),
\]
 where the second equality follows from Corollary \ref{cor:mindist}. 
\end{proof}

Lemma \ref{lem:generalbound} allows us to use existing bounds on the minimum distances of binary codes to bound $\exponent_\ell$:
\begin{example}[Sphere Packing Bound]
Applying the sphere packing bound for $d(\ell,\ell-i+1)$ in Lemma
\ref{lem:generalbound},
we get
\begin{align}\label{eqn:hamming}
\exponent_\ell \leq \frac 1\ell \sum_{i=1}^\ell \log_{\ell} \tilde D_i,
\end{align}
where 
\begin{align*}
\tilde D_i = \max\left\{ D: \sum_{j=0}^{\lfloor \frac{D-1}{2}\rfloor} \binom{\ell}{j} \leq 2^{i-1}
\right\}.
\end{align*}
\end{example}
Note that for small values of $n$ for which $d(n,k)$ is known for all $k\leq n$, the bound in Lemma \ref{lem:generalbound} can be evaluated exactly.

\subsection{Improved Upper Bound}
Bounds given in Section \ref{sub:upperbound} relate the partial distances $\{D_i\}$ to minimum distances of linear codes, but are loose since they do not exploit the dependence 
among the $\{D_i\}$. 
In order to improve the upper bound we use the following parametrization:
Consider an $\ell \times \ell$ matrix $G= [g_1^T,\dots,g_\ell^T]^T$. Let 
\begin{align*}
T_i & = \{k: g_{ik} = 1, g_{jk} = 0 \textrm{ for all } j>i\}\\
S_i & = \{k: \exists j>i \textrm { s.t.\ } g_{jk} = 1 \},
\end{align*}
and let $t_i = |T_i|$.
\begin{example}
For the matrix
\begin{align*}
F = \left[
\begin{array}{cccc}
0 & 0 & 0 & 1 \\
0 & 1 & 1 & 0 \\
1 & 1 & 0 & 0 \\
1 & 0 & 0 & 0 
\end{array}\right].
\end{align*}
$T_2 = \{3\}$ and $S_2 = \{1,2\}$.
\end{example}
Note that $T_i$ are disjoint and $S_i = \cup_{j=i+1}^\ell T_j$. Therefore, $|S_i| = \sum_{j=i+1}^\ell t_i$. Denoting the restriction of $g_j$ to the indices in $S_i$ by $g_{jS_i}$, we have
\begin{align}\label{eqn:D-split}
D_i = t_i + s_i, 
\end{align}
where $s_i \triangleq \dh(g_{iS_i}, \langle g_{(i+1)S_i}, \dotsc, g_{\ell S_i} \rangle)$. By a similar reasoning as in the proof of Lemma \ref{lem:monotonic}, it can be shown that there exists a matrix $G$ with 
\[
s_i \leq \dh(g_{jS_i},\langle g_{(j+1)S_i},\dotsc,g_{\ell S_i}\rangle) \quad \forall i<j,
%\dh(g_{iS_i}, \langle g_{(i+1)S_i}, \dotsc, g_{\ell S_i} \rangle) \leq \\ \dh(g_{jS_j}, \langle g_{(j+1)S_j}, \dotsc, g_{\ell S_j} \rangle) \quad\forall i<j,
\]
and
\[
\exponent(G) = \exponent_\ell.
\]
Therefore, for such a matrix $G$, we have (cf.\ proof of Lemma \ref{lem:generalbound})
\begin{align}\label{eqn:monotonic-s}
s_i \leq d(|S_i|, \ell-i+1).
\end{align}
Using the structure of the set $S_i$, we can bound $s_i$ further:
\begin{lemma}[Bound on Sub-distances]\label{lem:halfbound}
$s_i \leq \lfloor \frac{|S_i|}{2} \rfloor$.
\end{lemma}
\begin{proof}
We will find a linear combination of $\{g_{(i+1)S_i},\dotsc,g_{\ell S_i}\}$ whose Hamming distance to $g_{iS_i}$ is at most $\lfloor \frac{|S_i|}{2} \rfloor$. To this end define %the sets
%\begin{align*}
%V_m = \{j\in S_i: g_{mj} = 1 \textrm{ and } g_{kj} = 0\;\; \forall k > m\}.
%\end{align*}
%for $m = i+1,\dotsc,\ell$.
$w = \sum_{j=i+1}^\ell \alpha_j g_{jS_i}$, where $\alpha_j\in\{0,1\}$. Also
define $w_k = \sum_{j=i+1}^k \alpha_j g_{jS_i}$. Noting that the sets $T_j$s are
disjoint with $\cup_{j=i+1}^\ell T_j = S_i$, we have $\dh(g_{iS_i},w) = \sum_{j=i+1}^\ell \dh(g_{iT_j}, w_{T_j})$. 

We now claim that choosing the $\alpha_j$s in the order $\alpha_{i+1},\dotsc,\alpha_\ell$ by 
\begin{align}\label{eqn:alpha}
\argmin_{\alpha_j \in\{0,1\}} \dh(g_{iT_j},w_{{j-1}T_j}+\alpha_j g_{jT_j}),
\end{align}
we obtain $\dh(g_{iS_i},w)\leq \lfloor \frac{|S_i|}{2} \rfloor$. To see this, note that by definition of the sets $T_j$ we have $w_{T_{j}} = w_{jT_{j}}$. Also observe that by the rule (\ref{eqn:alpha}) for choosing $\alpha_j$, we have $\dh(g_{iT_j}, w_{jT_j}) \leq \lfloor \frac{|T_j|}{2} \rfloor$. Thus,
\begin{align*}
\dh(g_{iS_i},w) &= \sum_{j=i+1}^\ell \dh(g_{iT_j}, w_{T_j})\\
&=\sum_{j=i+1}^\ell \dh(g_{iT_j}, w_{jT_j}) \\
&\leq \sum_{j=i+1}^\ell \left\lfloor \frac{\vert T_j \vert}{2}\right\rfloor \leq \left\lfloor \frac{|S_i|}{2}\right\rfloor.
\end{align*}
\end{proof}

Combining (\ref{eqn:D-split}), (\ref{eqn:monotonic-s}) and Lemma \ref{lem:halfbound}, and noting that the invertibility of $G$ implies $\sum t_i = \ell$, we obtain the following:
\begin{lemma}[Improved Upper Bound]\label{lem:strongerupperbound}
\begin{align*}
\exponent_\ell \leq \max_{\sum_{i=1}^\ell t_i = \ell} \frac 1\ell \sum_{i=1}^\ell \log_\ell (t_i + s_i)%D_i \leq  t_i+\frac{1}{2}\sum_{j=i+1}^\ell t_j 
\end{align*}
where 
\begin{align*}
s_i =  \min \Bigl\{ \bigl\lfloor \frac 12 \sum_{j=i+1}^\ell t_j \bigr\rfloor, 
d\bigl(\sum_{j=i+1}^\ell t_j, \ell-i+1\bigr) \Bigr\}.
\end{align*}
\end{lemma}
The bound given in the above lemma is plotted in Figure~\ref{fig:bound}. 
It is seen that no matrix with exponent greater than $\frac 12$ can be found for $\ell \leq 10$. 

In addition to providing an upper bound to $\exponent_\ell$, Lemma \ref{lem:strongerupperbound} narrows down the search for matrices which achieve $\exponent_\ell$. 
In particular, it enables us to list all sets of possible partial distances with exponents greater than $\frac 12$. 
For $11 \leq \ell  \leq 14$, an exhaustive search for matrices with a ``good'' set of partial distances bounded by Lemma \ref{lem:strongerupperbound} 
(of which there are 285) shows that no matrix with exponent greater than $\frac 12$ exists.

\section{Construction Using BCH Codes}\label{sec:construction}
We will now show how to construct 
a matrix $G$ of dimension $\ell=16$ with exponent exceeding $\frac12$. In fact, we will show
how to construct the best such matrix. More generally, we will
show how BCH codes give rise to ``good matrices.'' 
Our construction of $G$ consists of taking an $\ell \times \ell$ binary
matrix whose $k$ last rows form a generator matrix of a $k$-dimensional BCH code.
%  so that its $k \times \ell$ dimensional sub-matrix consisting of the last $k$ rows of $G$
%is a generator matrix of a $k$-dimensional BCH code of length $\ell$. 
The partial distance $D_k$
is then at least as large as the minimum distance of this $k$-dimensional code.

To describe the partial distances explicitly we make use of the spectral view
of BCH codes as sub-field sub-codes of Reed-Solomon codes as described in \cite{Bla84}.
We restrict our discussion to BCH codes of length $\ell=2^m-1$, $m \in \naturals$.

Fix $m \in \naturals$. Partition the set of integers $\{0, 1, \dots, 2^m-2\}$
into a set ${\mathcal C}$ of {\em chords}, 
\begin{align*}
{\mathcal C} & = \cup_{i=0}^{2^m-2} \{2^k i \mod(2^m-1) : k \in \naturals\}.
\end{align*}
\begin{example}[Chords for $m=5$]\label{exa:listofchords}
For $m=5$ the list of chords is given by
\begin{align*}
{\mathcal C} = \{ & \{0\}, \{1, 2, 4, 8, 16\}, \{3, 6, 12, 17, 24\}, \\
& \{5, 9, 10, 18, 20\},  \{7, 14, 19, 25, 28\}, \\
& \{11, 13, 21, 22, 26\}, \{15, 23, 27, 29, 30\}\}. 
\end{align*}
\qed
\end{example}
Let $C$ denote the number of chords and assume that the chords
are ordered according to 
their smallest element as in Example~\ref{exa:listofchords}.
Let $\mu(i)$ denote the minimal element of chord $i$, $1 \leq i \leq C$
and let $l(i)$ denote the number of elements in chord $i$. Note that by this
convention $\mu(i)$ is increasing. It is well known that $1 \leq l(i) \leq m$ and
that $l(i)$ must divide $m$.
\begin{example}[Chords for $m=5$]\label{exa:chordproperties}
In Example \ref{exa:listofchords} we have $C=7$, $l(1)=1$, $l(2)=\dots=l(7)=5=m$,
$\mu(1)=0$, $\mu(2)=1$, $\mu(3)=3$, $\mu(4)=5$, $\mu(5)=7$, $\mu(6)=11$, $\mu(7)=15$.
\qed
\end{example}
Consider a BCH code of length $\ell$ and dimension $\sum_{j=k}^C l(j)$ for some $k\in\{1,\dotsc,C\}$. It is well-known that this code has minimum distance at least $\mu(k)+1$. Further, the generator matrix of this code is obtained by concatenating the generator matrices of two BCH codes of respective dimensions $\sum_{j=k+1}^C l(j)$ and $l(k)$. This being true for all $k\in\{1,\dotsc,C\}$, it is easy to see that the generator matrix of the $\ell$ dimensional (i.e., rate 1) BCH code, which will be the basis of our construction, has the property that its last $\sum_{j=k}^C l(j)$ rows form the generator matrix of a BCH code with minimum distance at least $\mu(k)+1$. 
This translates to the following lower
bound on partial distances $\{D_i\}$: Clearly, 
$D_i$ is least as large as the minimum distance of
the code generated by the last $\ell-i+1$ rows of the matrix. Therefore, if $\sum_{j=k+1}^{C} l(j) \leq \ell-i+1
\leq \sum_{j=k}^C l(j)$, then 
\[
D_{i} \geq \mu(k) +1.
\]
The exponent $\exponent$ associated with these partial design distances can then be bounded as
\begin{align}\label{eqn:exponentfromchord}
\exponent & \geq  \frac{1}{2^m-1}\sum_{i=1}^{C}  l(i) \log_{2^m-1}(\mu(i)+1).
\end{align}
\begin{example}[BCH Construction for $\ell=31$]\label{exa:misfive}
From the list of chords computed in Example~\ref{exa:listofchords} we obtain
\begin{align*}
\exponent & \geq  \frac{5}{31} \log_{31}(2 \cdot 4 \cdot 6 \cdot 8 \cdot 12 \cdot 16)
\approx 0.526433.
\end{align*}
An explicit check of the partial distances reveals that the above inequality is in fact an equality.
\qed
\end{example}

For large $m$, the bound in (\ref{eqn:exponentfromchord}) is not convenient to work with.
The asymptotic behavior of the exponent is however easy to assess by considering the following bound.
Note that no $\mu(i)$ (except for $i=1$) can be an even number since otherwise  
$\mu(i)/2$, being an integer, would be contained in chord $i$, a contradiction.
It follows that for the smallest exponent all chords (except
chord $1$) must be of length $m$ and that $\mu(i)=2i+1$.  
This gives rise to the bound
\begin{align}\label{eqn:bchexponent}
\exponent \geq &\frac{1}{(2^m-1)\log(2^m - 1)} \\
&\cdot 
\left(\sum_{k=1}^{a} m \log(2k) + (2^m-2 - a m) \log(2a+2)\right), \nonumber
\end{align}
where $a=\lfloor \frac{2^m-2}{m} \rfloor$.
It is easy to see that as $m \to \infty$ the above exponent tends to 1,
 the best exponent one can hope for (cf.\ Lemma~\ref{lem:exponent1}).
We have also seen in Example~\ref{exa:misfive} that for $m=5$ we achieve an exponent strictly
above $\frac12$.  

Binary BCH codes exist for lengths of the form $2^m-1$. To construct matrices of
other lengths, we use {\em shortening}, a standard method 
to construct good codes of smaller lengths from an existing code, which we recall here: Given a code
$\code{}$, fix a symbol, say the first one, and divide the codewords into two
sets of equal size
depending on whether the first symbol is a 1 or a 0. Choose the
set having zero in the first symbol and delete this symbol. The resulting
codewords form a linear code with both the length and dimension decreased by
one. The minimum distance of the resulting code is at least as large as the
initial distance. The generator matrix of the resulting code can be obtained
from the original generator matrix by removing a generator vector having a one
in the first symbol, adding this vector to all the remaining vectors starting
with a one and removing the first column.

Now consider an $\ell \times \ell$ matrix $G_\ell$. Find the column $j$ with the
longest run of zeros at the bottom, and let $i$ be the last 
row with a $1$ in this column. Then add the $i$th row to all the rows with a
$1$ in the $j$th column. Finally, remove the $i$th row and the $j$th column to obtain 
an $(\ell-1) \times (\ell-1)$ matrix $G_{\ell-1}$. 
The matrix $G_{\ell-1}$ satisfies the following property. 
\begin{lemma}[Partial Distances after Shortening]
Let the partial distances of $G_\ell$ be given by $\{D_1\leq \dots\leq D_\ell\}$. Let
$G_{\ell-1}$ be the resulting matrix obtained by applying the above shortening
procedure with the $i$th row and the $j$th column. Let the partial distances of $G_{\ell-1}$ be 
$\{D'_1,\dots, D'_{\ell-1}\}$. We have
\begin{align}
D'_k & \geq D_k, \quad 1 \leq k \leq i-1 \label{equ:partialdista}\\
D'_k &= D_{k+1}, \quad i \leq k \leq \ell-1. \label{equ:partialdistb}
\end{align}
\end{lemma}
\begin{proof}
Let $G_\ell=[g_1^T,\dots,g_\ell^T]^T$ and 
$G_{\ell - 1}=[{g'_1}^T,\dots, {g'_{\ell-1}}^T]^T$
For $i \leq k$, $g'_k$ is obtained by removing the $j$th column of $g_{k+1}$. Since
all these rows have a zero in the $j$th position their partial distances do not
change, which in turn implies \eqref{equ:partialdistb}. 

For $k\leq i$, note that the minimum distance of the code $\code'{} = \langle g'_{k}, \dots,
g'_{\ell-1}\rangle$ is obtained by shortening $\code{}= \langle g_{k},
\dots,g_{\ell}\rangle$.
Therefore, $D'_{k} \geq \dmin(\code'{}) \geq \dmin(\code{}) = D_k$. 
\end{proof}

\begin{example}[Shortening of Code]
Consider the matrix 
\begin{align*}
\left[
\begin{array}{ccccc}
1 & 0 & 1 & 0 & 1 \\
0 & 0 & 1 & 0 & 1 \\
0 & 1 & 0 & 0 & 1 \\
0 & 0 & 0 & 1 & 1 \\
1 & 1 & 0 & 1 & 1 
\end{array}\right].
\end{align*}
The partial distances of this matrix are $\{1, 2, 2, 2, 4\}$. 
According to our procedure, we pick the $3$rd column since it has 
 a run of three zeros at the bottom (which is maximal). We then add
the second row to the first row (since it also has a $1$ in the third column). 
Finally, deleting column $3$ and row $2$ we
obtain the matrix
\begin{align*}
\left[
\begin{array}{cccc}
1 & 0 & 0 & 0 \\
0 & 1 & 0 & 1 \\
0 & 0 & 1 & 1 \\
1 & 1 & 1 & 1 
\end{array}\right].
\end{align*}
The partial distances of this matrix are $\{1, 2, 2, 4\}$.
\qed
\end{example}

\begin{example}[Construction of Code with $\ell=16$]
Starting with the $31 \times 31$ BCH matrix and repeatedly applying the above
procedure results in the exponents listed in Table~\ref{tab:smallexponent}.
%\begin{table}[h]
%\begin{center}
%\begin{tabular}{c|c|c|c|c|c|c|c}
%$\ell$ & $\text{exponent}$ & $\ell$ & $\text{exponent}$ &  $\ell$ & $\text{exponent}$ &  $\ell$ & $\text{exponent}$\\  \hline
%$16$ & $0.51828$  &  $20$ & $0.49659$  &  $24$ & $0.50445$ &  $28$ & $0.51457$   \\ 
%$17$ & $0.49175$  &  $21$ & $0.48705$  &  $25$ & $0.50040$ &  $29$ & $0.51710$   \\  
%$18$ & $0.48968$  &  $22$ & $0.49445$  &  $26$ & $0.50470$ &  $30$ & $0.52205$   \\     
%$19$ & $0.48742$  &  $23$ & $0.50071$  &  $27$ & $0.50836$ &  $31$ & $0.52643$   \\     
%\end{tabular}
%\vspace{0.1in}
%\caption{\label{tab:smallexponent} The best exponents achieved by shortening the
%BCH matrix of length 31.}
%\end{center}
%\end{table}

\begin{table}[h]
\begin{center}
\begin{tabular}{c|c|c|c|c|c|c|c}
$\ell$ & $\text{exponent}$ & $\ell$ & $\text{exponent}$ &  $\ell$ & $\text{exponent}$ &  $\ell$ & $\text{exponent}$\\  \hline
$31$ & $0.52643$  &  $27$ & $0.50836$  &  $23$ & $0.50071$ &  $19$ & $0.48742$   \\ 
$30$ & $0.52205$  &  $26$ & $0.50470$  &  $22$ & $0.49445$ &  $18$ & $0.48968$   \\  
$29$ & $0.51710$  &  $25$ & $0.50040$  &  $21$ & $0.48705$ &  $17$ & $0.49175$   \\     
$28$ & $0.51457$  &  $24$ & $0.50445$  &  $20$ & $0.49659$ &  $16$ & $0.51828$  \\     
\end{tabular}
\vspace{0.1in}
\caption{\label{tab:smallexponent} The best exponents achieved by shortening the
BCH matrix of length 31.}
\end{center}
\end{table}

The $16 \times 16$ matrix having an exponent $0.51828$ is
\begin{align*}
\left[
\begin{array}{cccccccccccccccc}
1 & 0 & 0 & 1 & 1 & 1 & 0 & 0 & 0 & 0 & 1 & 1 & 1 & 1 & 0 & 1 \\  
0 & 1 & 0 & 0 & 1 & 0 & 0 & 1 & 0 & 1 & 1 & 1 & 0 & 0 & 1 & 1 \\
0 & 0 & 1 & 1 & 1 & 1 & 1 & 0 & 0 & 1 & 1 & 0 & 1 & 1 & 1 & 0 \\  
0 & 1 & 0 & 1 & 0 & 1 & 1 & 0 & 1 & 0 & 1 & 0 & 0 & 0 & 0 & 0 \\  
1 & 1 & 1 & 1 & 0 & 0 & 0 & 0 & 0 & 0 & 1 & 0 & 1 & 1 & 0 & 1 \\  
0 & 0 & 1 & 0 & 0 & 1 & 0 & 1 & 0 & 1 & 0 & 0 & 0 & 1 & 1 & 0 \\  
0 & 0 & 1 & 0 & 0 & 0 & 0 & 0 & 0 & 1 & 1 & 1 & 0 & 0 & 0 & 0 \\  
0 & 1 & 0 & 1 & 1 & 1 & 0 & 0 & 1 & 0 & 1 & 1 & 0 & 0 & 1 & 0 \\ 
1 & 1 & 1 & 0 & 0 & 1 & 1 & 0 & 1 & 0 & 0 & 1 & 0 & 1 & 0 & 0 \\ 
1 & 0 & 1 & 0 & 1 & 0 & 1 & 1 & 1 & 0 & 1 & 1 & 0 & 1 & 0 & 1 \\ 
1 & 1 & 1 & 0 & 0 & 0 & 0 & 0 & 0 & 0 & 0 & 1 & 1 & 0 & 1 & 0 \\ 
1 & 0 & 0 & 1 & 1 & 0 & 0 & 0 & 0 & 1 & 0 & 1 & 1 & 0 & 1 & 1 \\ 
1 & 1 & 1 & 1 & 1 & 0 & 1 & 0 & 0 & 0 & 0 & 1 & 0 & 1 & 0 & 0 \\ 
1 & 0 & 1 & 0 & 1 & 1 & 1 & 1 & 0 & 1 & 0 & 0 & 0 & 0 & 0 & 1 \\ 
1 & 0 & 1 & 0 & 0 & 0 & 0 & 1 & 0 & 1 & 1 & 1 & 1 & 1 & 0 & 0 \\ 
1 & 1 & 1 & 1 & 1 & 1 & 1 & 1 & 1 & 1 & 1 & 1 & 1 & 1 & 1 & 1 \\ 
\end{array}\right].
\end{align*}
The partial distances of this matrix are $\{16,8,8,8,8,6,6,4,4,4,4,2,2,2,2,1\}$.
Using Lemma~\ref{lem:strongerupperbound} we observe that for the $16\times 16$ case there are only $11$
other possible sets of partial distances which have a better exponent than the above matrix. An
exhaustive search for matrices with such sets of partial distances confirms that no such matrix exists.
Hence, the above matrix achieves the best possible exponent among all $16 \times 16$ matrices.
\qed
\end{example}

\section*{Acknowledgment}
We would like to thank E. Telatar for helpful discussions and his suggestions
for an improved exposition. This work was partially
supported by the National
Competence Center in Research on Mobile Information and Communication
Systems (NCCR-MICS), a center supported by the Swiss National Science
Foundation under grant number 5005-67322.

\appendix
In this section we prove the following lemma which is used in the proof of Lemma~\ref{lem:I-converges}.
\begin{lemma}[Mutual Information of $W^k$]
\label{lem:combinedI}
Let $W$ be a symmetric B-DMC and let $W^k$ denote the channel
\begin{align*}
W^k(y_1^k\mid x) = \prod_{i=1}^k W(y_i\mid x).
\end{align*}
If $I(W) \in (\delta, 1-\delta)$ for some $\delta > 0$, then there exists an $\eta (\delta)>0$ such that $I(W^k) - I(W) > \eta (\delta)$.
\end{lemma}
The proof of Lemma~\ref{lem:combinedI} is in turn based on the following theorem.
\begin{theorem}[\cite{SSZ05, LHHH05} Extremes of Information Combining]\label{thm:extremes}
Let $W_1,\dots,W_k$ be $k$ symmetric B-DMCs with capacities $I_1,\dots,I_k$
respectively. Let $W^{(k)}$ denote the channel with transition probabilities
\begin{align*}
W^{(k)}(y_1^k\mid x) = \prod_{i=1}^k W_i(y_i\mid x).
\end{align*}

Also let $W_{\BSC}^{(k)}$ denote the channel with transition probabilities
\begin{align*}
W_{\BSC{}}^{(k)}(y_1^k\mid x) =
\prod_{i=1}^k W_{\BSC{(\epsilon_i)}}(y_i\mid x),
\end{align*} 
where $\BSC(\epsilon_i)$ denotes the binary symmetric channel (BSC) with crossover probability  $\epsilon_i \in [0,\frac{1}{2}]$, $\epsilon_i \triangleq h^{-1}(1-I_i)$, where $h$ denotes the binary entropy function. Then, $I(W^{(k)}) \geq I(W_{\BSC{}}^{(k)})$.
\end{theorem}

\begin{remark}
Consider the transmission of a single bit $X$ using
$k$ independent symmetric B-DMCs $W_1,\dots, W_k$ with capacities
$I_1,\dots,I_k$. Theorem \ref{thm:extremes} states
that over the class of all symmetric channels with given mutual informations, 
the mutual information between the input and the output vector is minimized when each of the
individual channels is a BSC. 
\end{remark}

{\em Proof of Lemma \ref{lem:combinedI}: }
Let
$\epsilon \in [0,\frac12]$ be the crossover probability of a BSC with
capacity $I(W)$, i.e., $\epsilon = h^{-1}(1-I(W))$. 
Note that for $k\geq 2$, 
\begin{align*}
I(W^k) \geq I(W^2) \geq I(W).
\end{align*}
By Theorem \ref{thm:extremes}, we have $I(W^2) \geq
I(W_{\BSC(\epsilon)}^2)$. 
%\begin{align*}
%W_{\BSC{}}(Y_1,Y_2\mid X) = W_{\BSC{(\epsilon_0)}}(Y_1\mid X)W_{\BSC{(\epsilon_0)}}(Y_2\mid X)
%\end{align*}
A simple computation shows that 
\begin{align*}
I(W_{\BSC(\epsilon)}^2) = 1+h(2\epsilon\bar{\epsilon})- 2 h(\epsilon).
\end{align*}
We can then write
\begin{align}
I(W^k) - I(W) & \geq I(W_{\BSC(\epsilon)}^2) - I(W) \notag \\
& = I(W_{\BSC(\epsilon)}^2) - I(W_{\BSC(\epsilon)}) \notag  \\
& = h(2\epsilon \bar{\epsilon})-h(\epsilon).
%& > \eta(\delta) \label{eqn:binary-entropy}
\end{align}
Note that $I(W) \in (\delta, 1-\delta)$ implies $\epsilon \in (\phi (\delta), \frac 12 - \phi(\delta))$ where $\phi(\delta)>0$, which in turn implies $h(2\epsilon \bar{\epsilon})-h(\epsilon) > \eta(\delta)$ for some $\eta(\delta) > 0$.
\qed

\begin{lemma}
Consider a symmetric B-DMC $W$. Let $P_e(W)$ denote the bit error probability of uncoded transmission under MAP decoding. Then,
\begin{align*}
P_e(W) \geq \frac 12 ( 1 - \sqrt{1 - Z(W)^2}).
\end{align*}
\end{lemma}

\begin{proof}
One can check that the inequality is satisfied with equality for BSC. It is also known that any symmetric B-DMC $W$ is equivalent to a convex combination of several, say $K$, BSCs where the receiver has knowledge of the particular BSC being used. Let $\{\epsilon_i\}_{i=1}^K$ and $\{Z_i\}_{i=1}^K$ denote the bit error probabilities and the Bhattacharyya parameter of the constituent BSCs. Then, $P_e(W)$ and $Z(W)$ are given by
\begin{align*}
P_e(W) = \sum_{i=1}^K \alpha_i \epsilon_i, \qquad Z(W) = \sum_{i=1}^K \alpha_i Z_i
\end{align*}
for some $\alpha_i > 0$, with $\sum_{i=1}^K \alpha_i = 1$. Therefore,
\begin{align*}
P_e(W) & = \sum_{i=1}^K \alpha_i \frac12 (1 - \sqrt{1-Z_i^2})\\
& \geq \frac 12 (1 - \sqrt{1 - (\sum_{i=1}^K \alpha_i Z_i)^2})\\
& = \frac 12 (1 - \sqrt{1 - Z(W)^2}),
\end{align*}
where the inequality follows from the convexity of the function $x \to 1 - \sqrt{1-x^2}$ for $x\in (0,1)$.
\end{proof}

\bibliographystyle{IEEEtran} 
\bibliography{/home2/korada/bib/lth,/home2/korada/bib/lthpub}
%\bibliography{lth,lthpub}
\end{document}

%% file: definition.tex
\definecolor{TODO}{rgb}{0.6,0.6,0.6} % TO DO!!!!

\definecolor{TOCHECK}{rgb}{0.8,0.8,0.8} % TO CKECK!!!!

\newtheorem{theorem}{Theorem}
\newcommand{\btheo}{\begin{theorem}}
\newcommand{\etheo}{\end{theorem}}
\newcommand{\bproof}{\begin{proof}}
\newcommand{\eproof}{\end{proof}}
\newtheorem{definition}[theorem]{Definition}
\newcommand{\bdefi}{\begin{definition}}
\newcommand{\edefi}{\end{definition}}
\newtheorem{fact}[theorem]{Fact}
\newcommand{\bprop}{\begin{fact}}
\newcommand{\eprop}{\end{fact}}
\newtheorem{corollary}[theorem]{Corollary}
\newcommand{\bcor}{\begin{corollary}}
\newcommand{\ecor}{\end{corollary}}
\newtheorem{example}[theorem]{Example}
\newcommand{\bex}{\begin{example}}
\newcommand{\eex}{\end{example}}
\newtheorem{lemma}[theorem]{Lemma}
\newcommand{\blemma}{\begin{lemma}}
\newcommand{\elemma}{\end{lemma}}
\newtheorem{remark}[theorem]{Remark}
\newcommand{\bremark}{\begin{remark}}
\newcommand{\eremark}{\end{remark}}
\newtheorem{conj}[theorem]{Conjecture}
\newcommand{\bconj}{\begin{conj}}
\newcommand{\econj}{\end{conj}}
\newtheorem{observation}[theorem]{Observation}
\newcommand{\naturals}{\ensuremath{\mathbb{N}}}

\newcommand{\argmin}{\ensuremath{\text{argmin}}}

\def\0{{\tt 0}} % Ex.: BPSK modulation => 0 is encoded into +1
\def\1{{\tt 1}} % Ex.: BPSK modulation => 1 is encoded into -1
\def\?{{\tt *}} % erasure symbol
\newcommand{\qed}{{\hfill \footnotesize $\blacksquare$}}
\renewcommand{\mid}{\,|\,}

\newcommand{\BSC}{\ensuremath{\text{BSC}}}

\newcommand{\code}[1]{\mathtt{C}^{#1}}
\newcommand{\exponent}{\mathtt{E}}
\newcommand{\dmin}{\mathtt{dmin}}

\renewcommand{\dh}{{{d_{H}}}}

%% file: ps/transform.tex
\setlength{\unitlength}{0.5bp}%
\begin{picture}(320,200)(0,0)
\put(0,0){\includegraphics[scale=0.5]{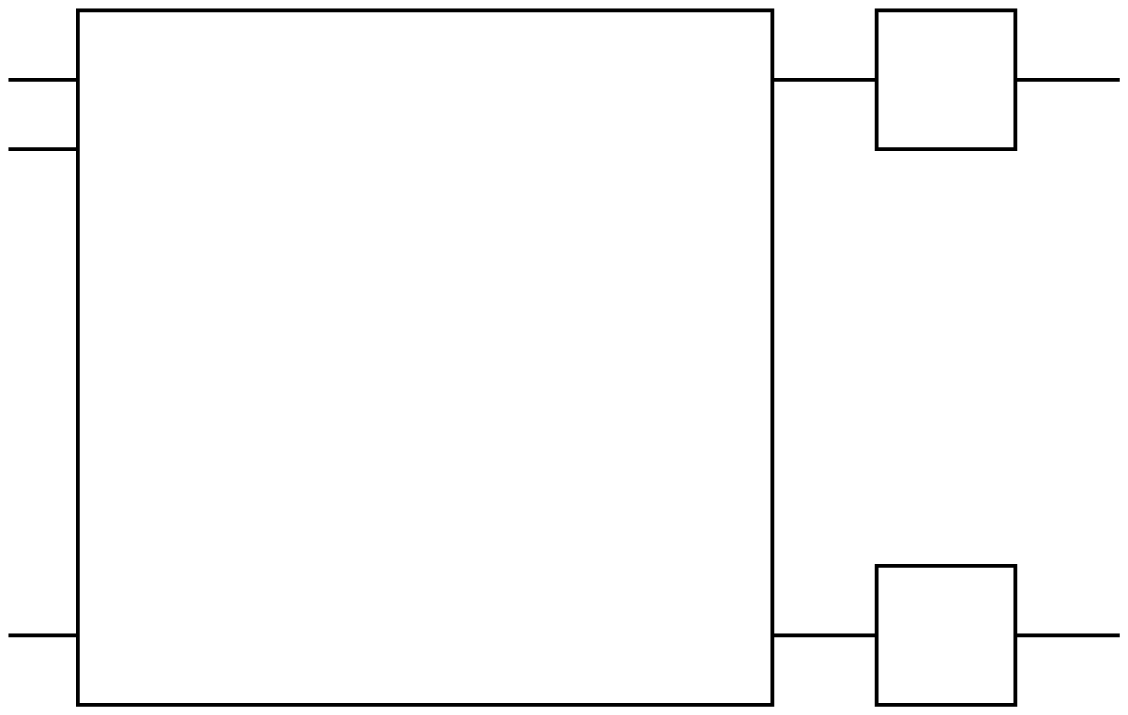}}
\put(262,185){\footnotesize $W$}
\put(265,145){\footnotesize$\cdot$}
\put(265,105){\footnotesize$\cdot$}
\put(265,65){\footnotesize$\cdot$}
\put(262,25){\footnotesize$W$}
\put(115,110){\footnotesize$G^{\otimes n}$}
\put(-40,185){\footnotesize $\textrm{bit}_1$}
\put(-40,165){\footnotesize $\textrm{bit}_2$}
\put(-30,130){\footnotesize $\cdot$}
\put(-30,95){\footnotesize $\cdot$}
\put(-30,60){\footnotesize $\cdot$}
\put(-40,25){\footnotesize $\textrm{bit}_N$}
\end{picture}

%% file: ps/geniechannel.tex
\setlength{\unitlength}{1bp}%
\begin{picture}(200,100)(30,0)
\put(0,0){\includegraphics[scale=1]{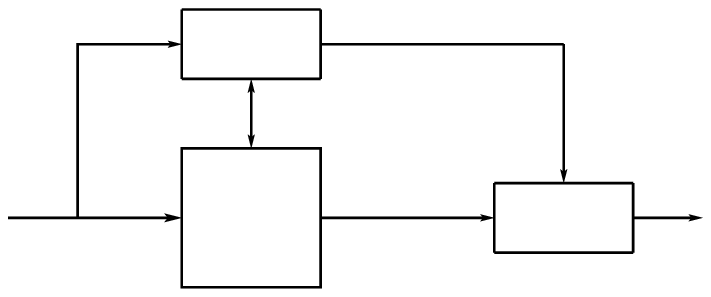}}
\put(95,27){\small $W_\ell$}
\put(92,77){\scriptsize Genie}
\put(175,28){\scriptsize Receiver}
\put(18,28){\small $u_i$}
\put(102,58){\footnotesize $0_1^{i-1}, u_{i+1}^\ell$}
\put(153,35){\small $y_1^\ell$}
\put(193,47){\small $u_{i+1}^\ell$}
\put(235,28){\small $\hat{u}_i$}
\end{picture}

%% file: ps/bounds.tex
\setlength{\unitlength}{1bp}%
\begin{picture}(215,128)(-15,-10)
\put(0,0){\includegraphics[scale=1.0]{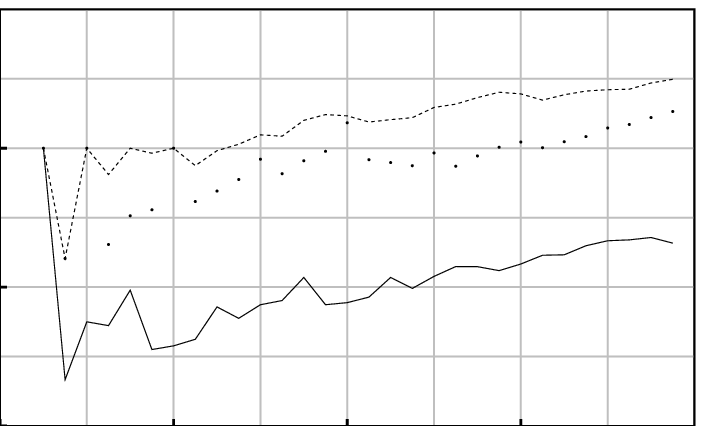}}
{
\footnotesize
\multiputlist(-15,40)(0,40)[l]{$0.4$,$0.5$}
\put(-15,0){\makebox(0,0)[lb]{$0.3$}}
\put(-15,120){\makebox(0,0)[lt]{$0.6$}}
\multiputlist(50,-10)(50,0)[b]{$8$,$16$,$24$}
\put(0,-10){\makebox(0,0)[lb]{$0$}}
\put(200,-10){\makebox(0,0)[rb]{$32$}}
}
\end{picture}